\documentclass[conference]{IEEEtran}
\IEEEoverridecommandlockouts
\usepackage{cite}
\usepackage{graphicx}
\usepackage{xcolor}
\usepackage{amsmath,amssymb,amsfonts}
\usepackage{amsthm}
\usepackage{dsfont}
\usepackage{enumitem}

\usepackage{url}

\usepackage[pscoord]{eso-pic}
\newcommand{\placetextbox}[3]{
  \setbox0=\hbox{#3}
  \AddToShipoutPictureFG*{
    \put(\LenToUnit{#1\paperwidth},\LenToUnit{#2\paperheight}){\vtop{{\null}\makebox[0pt][c]{#3}}}%
  }%
}%

\DeclareMathOperator*{\argmin}{arg\,min}

\def\BibTeX{{\rm B\kern-.05em{\sc i\kern-.025em b}\kern-.08em
    T\kern-.1667em\lower.7ex\hbox{E}\kern-.125emX}}

\newtheorem{theorem}{Theorem}
\newtheorem{lemma}{Lemma}

\newtheorem{defn}{Definition}
\newtheorem*{corollary}{Corollary}

\newcommand{\inv}[1]{{#1}^{-1}}

\begin{document}
\placetextbox{0.5}{0.98}{\texttt{In preparation for submission}}%

\title{Crowd Size Estimation for Non-Uniform Spatial Distributions with mmWave Radar}

\author{Anurag Pallaprolu, Aaditya Prakash Kattekola, Winston Hurst, \\ Upamanyu Madhow, Ashutosh Sabharwal and Yasamin Mostofi\thanks{Anurag Pallaprolu, Aaditya Prakash Kattekola, Winston Hurst, Upamanyu Madhow and Yasamin Mostofi are with the Dept. of Electrical and Computer Engineering, University of California, Santa Barbara, CA (email: \{apallaprolu, aadityaprakash, winstonhurst, madhow, ymostofi\} @ ucsb.edu). Ashutosh Sabharwal is with the Dept. of Electrical and Computer Engineering, Rice University, TX (email: ashu@rice.edu). This work is supported in part by NSF CNS award
2215646, and in part by ONR award N00014-23-1-2715.}}

\maketitle
\begin{abstract}
In this paper, we present a novel methodology for crowd size estimation using monostatic mmWave radar. Our aim is to accurately count large crowds that follow a non-uniform spatial distribution. Our estimation approach relies on the rigorous derivation of occlusion probabilities, which are then used to mathematically characterize the probability distributions that describe the number of agents visible to the radar as a function of the crowd size. We then estimate the true crowd size by comparing these derived mathematical models to the empirical distribution of the number of visible agents detected by the radar. This method requires minimal sensing capabilities (e.g., angle-of-arrival information is not needed), thus being well suited for either a dedicated mmWave radar or an integrated sensing and communication (ISAC) system. Extensive numerical simulations validate our methodology, demonstrating strong performance across diverse spatial distributions and for crowd sizes of up to (and including) 30 agents. We achieve a mean absolute error (MAE) of 0.48 agents, significantly outperforming a baseline which assumes that the agents are uniformly distributed in the area. Overall, our approach holds significant promise for a variety of applications including network resource allocation, crowd management, and urban planning.

\end{abstract}
\begin{IEEEkeywords}
Crowd Size Estimation, Integrated Sensing and Communication, Crowd Analytics, mmWave Radar\vspace{-5pt}
\end{IEEEkeywords}
\section{Introduction}
\label{sec:intro}

Driven in part by the use of mmWave spectrum in 5G networks~\cite{ericssonNov2023MobilityReport} and its continued importance in 6G~\cite{rappaport2019wireless}, the use of mmWave radar has garnered significant attention over recent years. Whether as standalone units or in the context of Integrated Sensing and Communication (ISAC)~\cite{liu2022integrated, yao2022intelligent}, sensing with mmWave signals can deliver critical contextual information that not only enhances communication system performance but also serves as a key resource for a wider range of other applications.

In many settings, a key piece of contextual information is the number of people in an area. In addition to being valuable in a communications context (\textit{e.g.,} traffic prediction), crowd analytics offer valuable insights into collective behaviors, which are essential for applications ranging from retail~\cite{cruz2020CCTVRetail} to public health \cite{korany2021counting}. However, crowd counting using the mmWave spectrum faces significant challenges, as Line-of-Sight (LOS) blockages at higher frequencies result in the inability to detect individuals who are occluded by other people and objects~\cite{maltsev2009experimental}.

In this paper, we set forth a method for crowd size estimation capable of handling inhomogeneous spatial distributions, \textit{i.e.,} distributions where the density of agents varies across space due to external factors, examples of which are shown in Fig.~\ref{fig:motivation}. For instance, certain regions of an area may be inaccessible due to physical barriers (\textit{e.g.,} walls, chairs/tables), which make these regions untraversable. Inhomogeneity can also result from crowd clustering around landmarks and points of interest, such as promotional displays or food vendors.

\begin{figure}
    \centering
        \includegraphics[width=\linewidth]{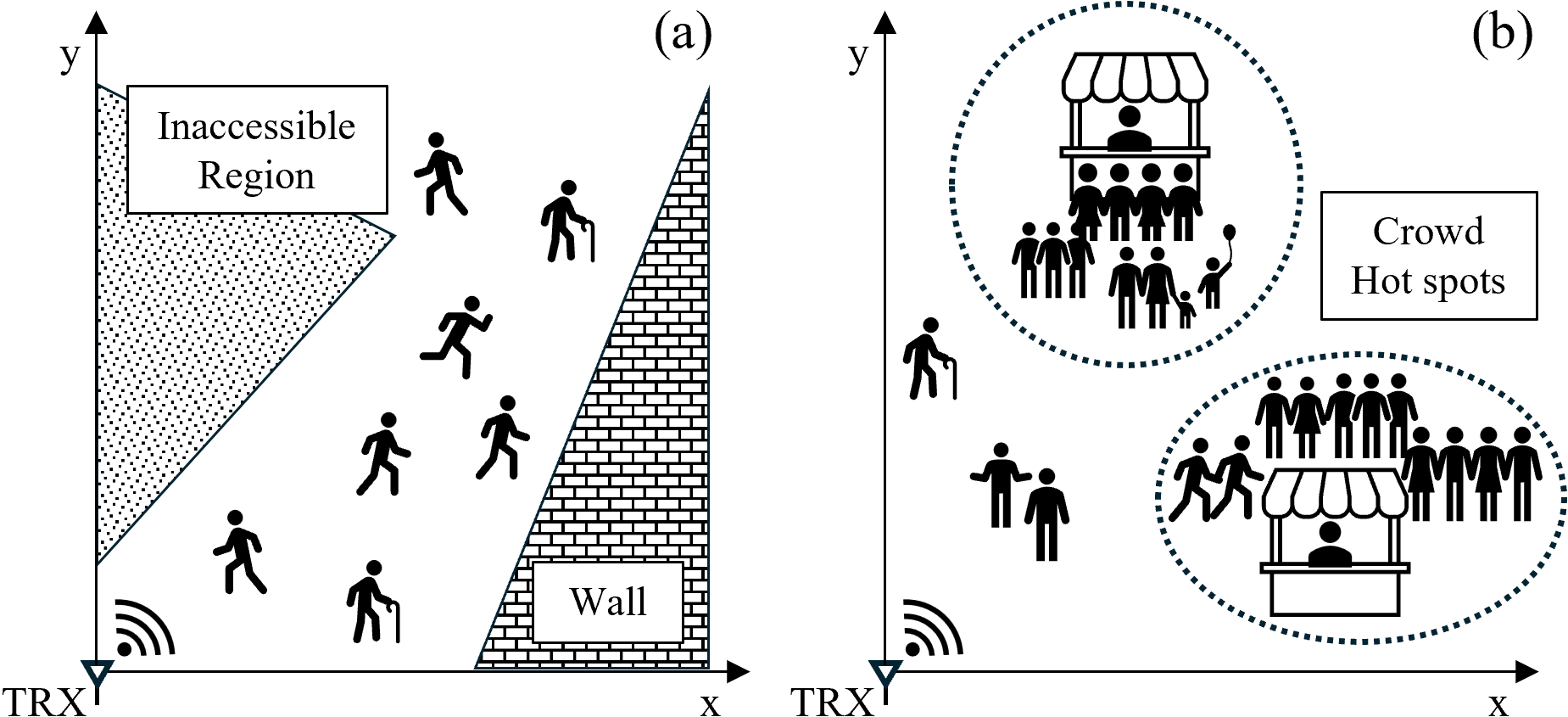}
        \vspace{-20pt}
        \caption{Examples of spatially inhomogeneous crowd behavior: (a) Regions of an area rendered inaccessible due to hard physical constraints, and (b) Crowds forming around points of interest and hotspots, leading to inhomogeneity.}
        \vspace{-20pt}
     \label{fig:motivation}
\end{figure}

\begin{figure*}
    \centering
        \includegraphics[width=\textwidth]{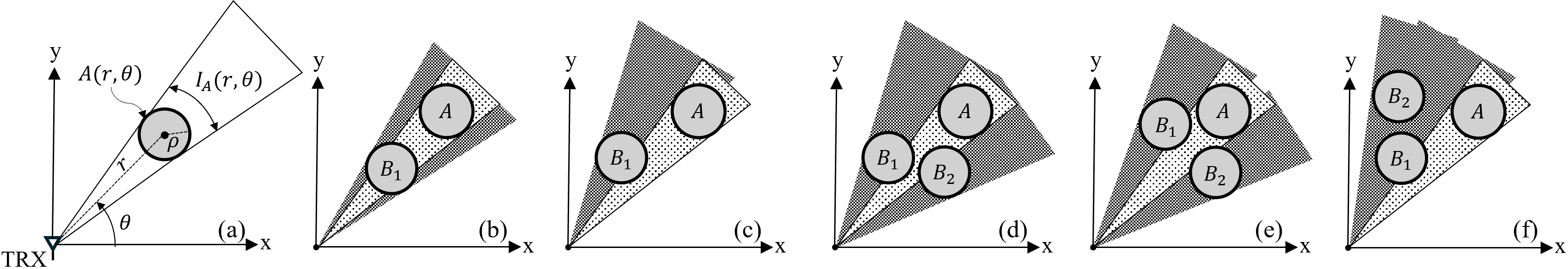}
        \vspace{-18pt}
        \caption{Blockage modeling of an arbitrary agent $A$ centered at $(r, \theta)$ in a crowd of $N$ mobile agents: (a) Geometrical definition of the visibility interval $I_A$ (b) $B_1$ blocks $A$ completely (c) $B_1$ blocks $A$ partially, rendering $A$ visible (d) $B_1$ and $B_2$ fully block $A$ simultaneously, while only partially blocking $A$ individually (e) $B_1$ and $B_2$ partially block $A$ together (f) $B_1$ and $B_2$ on one side of $A$.}
        \vspace{-18pt}
     \label{fig:vis-scenarios}
\end{figure*}

Our method represents a significant improvement in crowd size estimation over existing literature, as it can, more accurately, count larger crowds and effectively handle non-uniform crowd distributions. While WiFi signals have been previously used for occupancy estimation~\cite{depatla2015occupancy}, most prior works on mmWave crowd sensing have focused on counting based on tracking, and they are thus typically limited to handling small crowds (e.g., 6 or fewer individuals \cite{liu2024pmtrack, jiang2023robust}). More recently, we proposed a crowd size estimation framework using commodity mmWave FMCW radars, which accurately counted up to (and including) 21 people in real-world experiments \cite{pallaprolu2024}. However, that work assumes a uniformly distributed crowd, a condition that may not be true in some scenarios. Thus, this paper moves the state of the art forward by proposing a novel methodology for crowd size estimation that can handle non-uniform spatial distributions. More specifically, the spatial distribution of the crowd can be found from prior knowledge, such as the geometry of fixed barriers or historically observed crowd patterns around landmarks. We then demonstrate how this information can be leveraged to infer the total crowd size, even when some of the agents in the crowd are frequently not visible to the radar due to occlusion. We next discuss our contributions in detail.  

\noindent\textbf{Contributions:}
Estimating crowd sizes with mmWave signals presents challenges, as occluded individuals will not be visible to the radar board (thus not sensed), due to significant signal attenuation at these frequencies.  In addition, when considering non-uniform spatial distributions, the probability of visibility of an individual becomes location-dependent, posing new challenges that require new analysis.  In this paper, we build on our recent work \cite{pallaprolu2024} to develop a mathematical framework that can address non-uniform spatial distributions. More specifically, through rigorous mathematical modeling, we first characterize the location-dependent probability of blockage (visibility) of an individual.  We then show how to mathematically derive and utilize probability mass functions (PMFs) that describe the distribution of the number of visible agents as a function of the crowd size. These PMFs are then compared to the empirical distribution of the number of agents sensed by the radar to infer the crowd size. Extensive simulation results show that this approach works very well for spatially inhomogeneous crowds of up to (and including) 30 agents, achieving an MAE of 0.48 agents and further outperforming the state of the art. 

It is worth noting that our location-dependent mathematical modeling can also be used in a different context, \textit{i.e.,} to spatially characterize the statistics of LOS blockage, a problem of importance for channel modeling or adaptive beamforming in 5G/6G cellular networks~\cite{zeng2024tutorial}. In the existing literature, crowd modeling frequently assumes a homogeneous Poisson Point Process~\cite{hmamouche2021new}, which implies both spatial uniformity and point-like agents. However, this assumption can be limiting in capturing many practical scenarios. The proposed work of this paper can also contribute to this area as part of future work. 


\section{Blockage Modeling of Disc-Shaped Agents}
\label{sec:block_mod}
We consider a monostatic radar, either a standalone device or an ISAC-compatible communication access point (AP), with a field of view described by a quadrant $Q = \{(r, \theta)\ |\ r \in [\rho, r_m], \theta \in [0, \pi/2]\}$, reflecting the sectorized nature of mmWave APs. Consistent with the design of contemporary smart roadside infrastructure~\cite{agrawal2022intelligent}, the transmitting antenna height is taken to be around the human torso, so that the cylinder model is valid. Without loss of generality, we take the radar as the origin of our coordinate system. Our objective is to count the total number of agents, $N$, within $Q$. We model the agents as discs of radius $\rho$ with centers that are spatially distributed according to spatial probability distribution function denoted by $\mathbb{P}(\mathbf{x})$, which is defined for $x\in Q$ and known a priori. The radar is capable of detecting agents using techniques such as those presented in \cite{pallaprolu2024}, provided an LOS path exists from the agent to the radar. However, agents in the crowd may occlude each other, so that $N$ cannot be directly observed. In this section, we begin our derivation of the PMF of the number of visible agents, given a crowd of size $N$, by first presenting our mathematical model for blockages in the crowd.  

We next define the visibility interval of a disc-shaped agent $A$. Intuitively, any ray launched at an angle within this interval will necessarily intersect $A$, as shown in Fig.~\ref{fig:vis-scenarios} (a).

\begin{defn}
\label{def:vis_interval}
\textbf{Visibility Interval:} For any disc-shaped agent, A, centered at $(r, \theta)$ with a radius of $\rho$, we define its visibility interval as $I_A(r,\theta) = \biggl[\theta - \inv\sin\left(\frac{\rho}{r}\right)\!, \theta + \inv\sin\left(\frac{\rho}{r}\right)\biggr]$.
\end{defn}

We state that a Complete 1-Blockage occurs when an agent $A$ is entirely occluded by blocker $B_1$, as illustrated in Fig.~\ref{fig:vis-scenarios} (b), whereas a Partial 1-Blockage arises when $B_1$ only partially occludes $A$, as shown in Fig.~\ref{fig:vis-scenarios} (c). These concepts are formally presented in the following definitions.

\begin{defn}
\label{def:complete_one_blockage}
\textbf{Complete 1-Blockage:} Agent $A$ encounters a Complete 1-Blockage due to a blocker, $B_1$, only if $ I_{A} \subset I_{B_1}$.
\end{defn}

\begin{defn}
\label{def:partial_one_blockage}
\textbf{Partial 1-Blockage:} Let $G = I_A \cap I_{B_1}$, then $A$ encounters a Partial 1-Blockage only if $G \neq \phi$ and $G \subset I_A$.
\end{defn}

It is easy to see that $A$ encounters a Complete 1-Blockage due to $B_1$ only if $I_A \cap I_{B_1} = G = I_A$, and $B_1$ does not cast a shadow onto $A$ if $G = \phi$. In the case of a Partial 1-Blockage, it is possible for a different blocker, $B_2$, to also partially occlude $A$, such that the combined effect of $B_1$ and $B_2$ results in the complete blockage of $A$, as depicted in Fig~\ref{fig:vis-scenarios} (d). We next formally define this in terms of visibility intervals.

\begin{defn}
\label{def:complete_two_blockage}
\textbf{Simultaneous 2-Blockage:} Agent $A$ encounters a Simultaneous 2-Blockage due to blockers $B_1$ and $B_2$ only if
\begin{align}
    \label{eq:complete-2-blockage-def}
    (I_A \subset ( I_{B_1} \cup I_{B_2}))\ \wedge\ ( (I_A \not\subset I_{B_1}) \ \wedge\  (I_A \not\subset I_{B_2}))
\end{align}
\end{defn}
\begin{lemma}
\label{lemma:complete_2_blocker}
A single agent $A$ cannot encounter a Simultaneous 2-Blockage by blockers $B_1$ and $B_2$ if $I_{B_1} \cap I_{B_2} = \phi$.
\end{lemma}
\begin{proof}\renewcommand{\qedsymbol}{}
Consider two blockers $B_1$ and $B_2$. As shown in Fig.~\ref{fig:vis-scenarios} (e), if their visibility intervals do not overlap, then $I_{B_1} \cap I_{B_2} = \phi$. As $A$ experiences a simultaneous blockage, we have
\begin{align}
    I_A &\nsubseteq  (( I_{B_1} \cup I_{B_2}) - (I_{B_1} \cap I_{B_2}))  = (( I_{B_1} \cup I_{B_2}) - \phi) \nonumber \\
    & \implies I_A \nsubseteq  ( I_{B_1} \cup I_{B_2}),\nonumber
\end{align}
which contradicts Eq.~\ref{eq:complete-2-blockage-def}. Thus, we have shown that $A$ cannot encounter a Simultaneous 2-Blockage if $I_{B_1} \cap I_{B_2} = \phi$.
\end{proof}
\vspace{-7pt}
We now formalize the fact that agent $A$ cannot be simultaneously blocked if all the blockers lie on the same side of $A$, as evidenced by the scenario shown in Fig.~\ref{fig:vis-scenarios} (f).

\begin{lemma}
\label{lemma:same_side}
A single agent $A$ located at $(r_A, \theta_A)$ cannot encounter a simultaneous blockage by $N \geq 2$ blockers $B_i$, located at $(r_{B_i}, \theta_{B_i})$, if $\theta_A > \theta_{B_i}\ \forall\ i$ or $\theta_A < \theta_{B_i}\ \forall\ i$ \textit{i.e.,} if all $B_i$ lie on the same side of $A$.
\end{lemma}
\begin{proof}\renewcommand{\qedsymbol}{}
We begin by considering $N$ blockers located at $(r_{B_i}, \theta_{B_i})\ \forall\ i = 1, 2, \hdots, N$ and an agent $A$ located at $(r_A, \theta_A)$ such that all the blockers are on the same side. We then define $G_i = I_A \cap I_{B_i}$ (see Def.~\ref{def:partial_one_blockage}). Next, let us order these blockers based on the extent of their partial blockage on $A$. Without loss of generality, we have:
\begin{align}
\label{eq:partial_blockage_ordering}
|I_A| > |G_1| \geq |G_2| \geq\hdots \geq |G_N| > 0.
\end{align}
We next show that there always exists a dominant partial blocker rendering all other blockers redundant.\\
\noindent$\textbf{When } \pmb{\theta_A > \theta_i\ :}$ As $G_i \neq I_A \ \text{and}\ G_i \neq \phi\ \forall\ i$, in this case we can directly see that $G_i$ must be of the form $G_i = \left[\theta_A -\inv\sin\left(\frac{\rho}{r_A}\right), g_i\right], \text{ for suitable }g_i \in \mathbb{R}$.
Applying the ordering introduced in Eq.~\ref{eq:partial_blockage_ordering}, we see that $|G_1| \geq\hdots \geq |G_N| \implies  g_1 \geq \hdots \geq g_N$, which is sufficient to prove that $G_1 \supseteq G_2 \supseteq \hdots \supseteq G_N$.\\
\noindent\textbf{When } $\pmb{\theta_A < \theta_i\ :}$ By applying symmetry, it is easy to show that even in this case $G_1 \supseteq G_2 \supseteq \hdots \supseteq G_N$. 

In conclusion, we see that $G_1$ contains all other partial blockages, and we have shown that when all blockers are on the same side, $B_1$ is the dominant partial blocker.
\end{proof}
We shall next define a Simultaneous 3-Blockage using the visibility interval formalism.
\begin{defn}
\label{def:simultaneous_3_blockage}
\textbf{Simultaneous 3-Blockage:} Agent $A$ experiences a Simultaneous 3-Blockage due to blockers $B_1, B_2$ and $B_3$ if
\begin{align*}
    &((I_A \subset ( I_{B_1} \cup I_{B_2} \cup I_{B_3}))\ \wedge\ I_A \not\subset (I_{B_1} \cup I_{B_2})\ \\  
    & \wedge\ I_A \not\subset (I_{B_2} \cup I_{B_3})\ \ \wedge\ 
     I_A \not\subset (I_{B_3} \cup I_{B_1})\ \\ & \wedge\  I_A \not\subset I_{B_1}\ \wedge\  I_A \not\subset I_{B_2}\ \wedge\  I_A \not\subset I_{B_3}.
\end{align*}
\end{defn}
\begin{theorem}
\label{thm:sim_3_blockage}
A single agent, $A$, cannot encounter a Simultaneous 3-Blockage by blockers $B_1, B_2$ and $B_3$, in a crowd of disc-shaped agents.
\end{theorem}
\begin{proof}\renewcommand{\qedsymbol}{}
Shown using Lemma 2, the proof is omitted for brevity.
\end{proof}

\begin{corollary}
\label{corr:simultaneous-k-blockage}
A single agent $A$ cannot encounter a Simultaneous K-Blockage by blockers $B_1, B_2,\hdots B_K$ for $K \geq 3$, in a crowd of disc-shaped agents.
\end{corollary}
In the next section, we apply our blockage model to develop a mathematical framework for crowd size estimation.

\section{A Mathematical Framework for Crowd Size Estimation}\label{sec:vis_maps}
We now introduce a crowd size estimation method based on the blockage model developed in Section~\ref{sec:block_mod} and  further extend the approach presented in \cite{pallaprolu2024} to address non-uniform spatial crowd distributions. We begin by mathematically modeling the blockage probability at a given point in space, which we then use to characterize the distribution of the total number of visible agents for a given crowd size $N$. Throughout, we define $\mathcal{V}$ as the event in which an agent, located at $\mathbf{x} = (r\cos\theta, r\sin\theta)$, is visible to the AP, and denote the likelihood of visibility by $\mathbb{P}(\mathcal{V}|N, \mathbf{x})$.

Define $D_1$ and $D_2$ as the events that $A$ experiences at least one Complete $1-$Blockage and at least one Simultaneous $2-$Blockage, respectively. We analyze these events using our interval-based visibility framework, which helps us compute the likelihood of visibility as
\vspace{-3pt}
\begin{align}
\label{eq:vis_n_x}
    &\mathbb{P}(\mathcal{V} | N, \mathbf{x}) = 1 - \mathbb{P}(D_1\cup D_2|N, \mathbf{x}) \nonumber \\
    &= 1 - \mathbb{P}(D_1|N, \mathbf{x}) - \mathbb{P}(D_2|N, \mathbf{x}) + \mathbb{P}(D_1 \cap D_2 | N, \mathbf{x}).
\end{align}

\vspace{-3pt}
We now derive analytical expressions for $\mathbb{P}(D_1|N, \mathbf{x})$, $\mathbb{P}(D_2|N, \mathbf{x})$, and $\mathbb{P}(D_1\cap D_2|N, \mathbf{x})$, thereby completing the characterization of $\mathbb{P}(\mathcal{V}|N, \mathbf{x})$. Using this, we further determine the network size, $N$, solely based on the statistics of the observed number of visible agents. The Principle of Inclusion and Exclusion~\cite{comtet2012advanced} is applied extensively to compute various blockage-related densities, which we present below for the sake of completeness.
\begin{theorem}
\label{thm:IEP}
\textbf{Principle of Inclusion and Exclusion (PIE): }Let $C_1, C_2,\hdots , C_M$ be events in a probability space $(\Omega, \mathcal{F}, \mathbb{P})$. Let us define the generic summand
\vspace{-5pt}
\begin{align}
\label{eq:s_k_blockage}
S_k=\sum_{1\leq i_1 < i_2<\hdots<i_k\leq M-1}\!\!\!\!\!\!\!\!\!\mathbb{P}\!\left(C_{i_1}\cap C_{i_2}\cap \hdots \cap C_{i_k}\right).
\end{align}
We then have $\mathbb{P}\left(C_1 \cup C_2 \cup \hdots \cup C_M\right)\!=\!\sum_{k=1}^{M}(-1)^{k+1}S_k.$
\end{theorem}

\subsection{$\mathbb{P}(D_1|N,\mathbf{x})$: At Least One Complete 1-Blockage}
We begin by analyzing the probability that the agent $A$ (located at $\mathbf{x}\in Q$) encounters at least one Complete 1-Blockage. We define $O_i$ as the event that $A$ experiences a Complete 1-Blockage due to the $i^{\text{th}}$ blocker, $B_i$, where $1\leq i \leq N-1$. We recall that all agent locations in the crowd are independently sampled from the prior spatial distribution, 
$\mathbb{P}(\mathbf{x})$. Given that the event $O_i$ is fully determined by the visibility interval of $B_i$ (Def.~\ref{def:complete_one_blockage}), we conclude that the events $O_i\sim O$ are also independent and identically distributed (i.i.d), when conditioned on $A$ being positioned at $\mathbf{x}$. Thus, we utilize the PIE to obtain $\mathbb{P}(D_1|N,\mathbf{x}) = \mathbb{P}(O_1 \cup O_2 \cup \hdots \cup O_{N-1} | N, \mathbf{x})$:
\begin{align}
\label{eq:iep_d1}
    \mathbb{P}(D_1|N,\mathbf{x}) =\!\sum_{k=1}^{N-1}(-1)^{k+1}\binom{N-1}{k}\mathbb{P}\!\left(O|N,\mathbf{x}\right)^k.
\end{align}
In order to evaluate $\mathbb{P}(O|N, \mathbf{x})$, we utilize Def.~\ref{def:complete_one_blockage}. We drop the conditional dependence\footnote{The event $O$ only relies on the fact that $N \geq 2$, which is taken care of by the binomial coefficient in the summation of Eq.~\ref{eq:iep_d1}.} on $N$, and we require the center of the disc-shaped blocker $B_1$ to lie in a region such that $I_A \subset I_{B_1}$ (see Fig.~\ref{fig:vis-scenarios} (b)). In other words,
\begin{align}
\label{eq:s_k_exact_1_blockage}
 \mathbb{P}(O|\mathbf{x}) = p_1(\mathbf{x}) \overset{\Delta}{=} \int_{\mathcal{R}_1(\mathbf{x})} \mathbb{P}(\mathbf{y})d\mathbf{y},
\end{align}
where $\mathcal{R}_1(\mathbf{x}) = \{\mathbf{y} \in Q\mid I(\mathbf{y}) \supset I(\mathbf{x})\}$, with $I(\mathbf{x})$ as defined in Def.~\ref{def:vis_interval}. Using $\mathcal{R}_1(\mathbf{x})$, we can then exactly obtain $p_1(\mathbf{x})$ by evaluating the integral of Eq.~\ref{eq:s_k_exact_1_blockage}. We plug the result of this integration into Eq.~\ref{eq:iep_d1} to obtain $\mathbb{P}(D_1|N,\mathbf{x}) = 1 - (1 - p_1(\mathbf{x}))^{N-1}$. We next provide an analytical derivation of $\mathcal{R}_1(\mathbf{x})$.
\begin{theorem}
\label{thm:single_onion}
Let a disc-shaped agent $A$ of radius $\rho$ be located at a point $\mathbf{x} = (r_A, \theta_A) \in Q$. The region $\mathcal{R}_1(\mathbf{x})$ is given by $\{(r, \theta)\mid \rho \leq r < r_A, \theta_A - \theta_c(r) \leq \theta \leq \theta_A + \theta_c(r)\}$, where 
\begin{align}
    \theta_c(r) = \inv\tan\left(\frac{\rho}{r} - \frac{\rho}{\sqrt{r_A^2 - \rho^2}}\right) \nonumber
\end{align}
\end{theorem}

\begin{proof}\renewcommand{\qedsymbol}{}
Consider the scenario shown in Fig.~\ref{fig:single_onion} (a), wherein the center of agent $A$ is denoted by $C_1(r_A, \theta_A)$, and the tangents to $A$ from the origin $O$ meet it at $P$ and $P'$. Using simple trigonometry, it can be seen that $\angle P'C_1 P = 2\inv\cos(\rho/r_A)$, and using the cosine rule, we obtain $PP' = 2\rho\sqrt{r_A^2 - \rho^2}/r_A$. Let us next consider an arbitrary blocker $B_1$, located at a distance $r$ from the TRX, such that it leads to a Complete 1-Blockage of $A$. From Fig.~\ref{fig:single_onion} (a), we see that $\triangle QOQ' \sim \triangle POP'$, and as a result, we have $OH/OC_1 = QQ'/PP'$. This gives us the length of $QQ' = 2\rho r/\sqrt{r_A^2 - \rho^2}$. We gauge the extent of occlusion by considering the portion of $QQ'$ not contained in $B_1$ as the margin of visibility. Fig.~\ref{fig:single_onion} (b) depicts a scenario where the blocker, $B_1$, with center at $C_2(r, \theta)$, partially occludes $A$ (not shown). Thus, $A$ is not completely blocked only if it has a positive margin of visibility \textit{i.e.,} $Q'D > 0$. It can then be seen that\footnote{Technically $HD = TD - TH$, where $T$ is the foot of the perpendicular dropped from $C_2$ onto $QQ'$. We assume $C_2T$ is small to simplify analysis.} 
\begin{align}
\label{eq:geom_single_onion}
Q'D = HQ' - HD \approx HQ' - (C_2D - C_2H)
\end{align}
From symmetry, it can be seen that $HQ' = QQ'/2 = \rho r/\sqrt{r_A^2 - \rho^2}$, and $C_2D = \rho$ by definition. Considering the $\triangle OC_2H$, we further conclude that $C_2H \approx r\tan\theta_c$. Substituting these into Eq.~\ref{eq:geom_single_onion}, and solving for $\theta_c$, we obtain the desired result.
\end{proof}

\begin{figure}
    \centering
        \includegraphics[width=\linewidth]{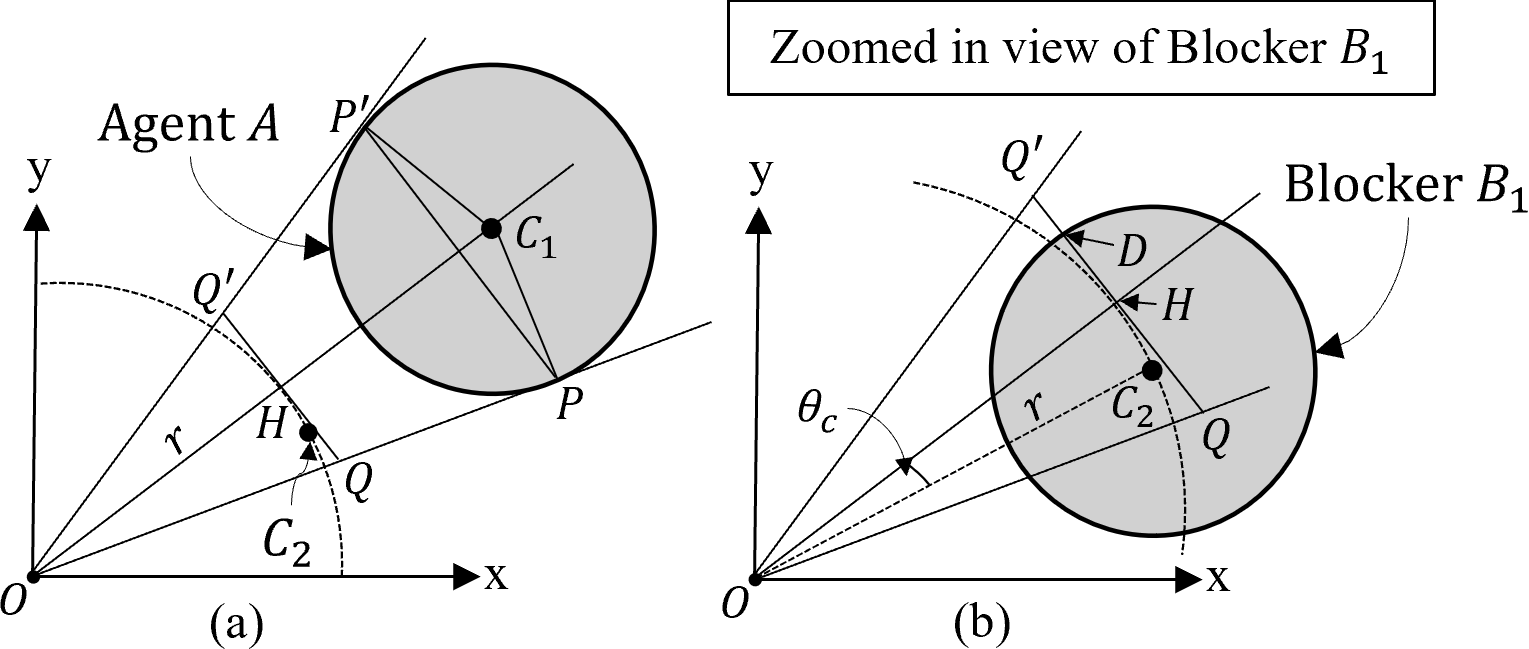}
        \vspace{-20pt}
        \caption{Geometry of a Complete 1-Blockage: (a) $QQ'$ defines the maximum margin of visibility for an agent $A$, centered at $C_1$, encountering a Complete 1-Blockage due to a blocker $B_1$ at radius $r$, centered at $C_2$. (b) Positive margin of visibility $Q'D > 0$ due to $B_1$ at radius $r$, centered at $C_2$.} 
        \vspace{-15pt}
     \label{fig:single_onion}
\end{figure}

\subsection{$\mathbb{P}(D_2|N,\mathbf{x})$: At Least One Simultaneous 2-Blockage}
We next analyze the probability that the agent $A$ (still located at $\mathbf{x}$) experiences at least one Simultaneous 2-Blockage, as introduced in Def.~\ref{def:complete_two_blockage}. We define $T_i$ as the event that $A$ encounters a Simultaneous 2-Blockage due to the $i^{\text{th}}$ pair of blockers, where $1 \leq i \leq L$, and $L = \binom{N-1}{2}$ is the total number of unique blocker pairs. Thus, we have
\begin{align}
\label{eq:iep_d3}
    \mathbb{P}(D_2|N,\mathbf{x}) = \mathbb{P}(T_1 \cup T_2 \cup \hdots \cup T_L | N, \mathbf{x}).
\end{align}
We shall now argue that the events $T_i$ are identically distributed when conditioned on $\mathbf{x}$ but are not necessarily independent. We see from Def.~\ref{def:complete_two_blockage} and Lemmas~\ref{lemma:complete_2_blocker} \&~\ref{lemma:same_side} that the pair of blockers $(B_1, B_2)$ causing the Simultaneous 2-Blockage are fully determined by their visibility intervals and a fixed inter-dependence structure (both of them located on the opposite sides of $A$, and $I_{B_1} \cap I_{B_2} \neq \phi$). 
Since $B_1$ and $B_2$ are identically distributed for a given $A$, the joint distribution of the pair $(B_1, B_2)$ responsible for the Simultaneous 2-Blockage is the same across all such blocker pairs when conditioned on $A$ being positioned at $\mathbf{x}$. 

To show that the events $T_i$ are not independent, consider $T_i$ and $T_j$ to be two events that correspond to $A$ experiencing Simultaneous 2-Blockages by the blocker pairs $(B_{1i}, B_{2i})$ and $(B_{1j}, B_{2j})$ respectively, where $B_{1i}$ represents the first blocker in the $i^{\text{th}}$ blocker pair. If the pairs share a common element, e.g., $B_{2i} = B_{1j}$, then $\mathbb{P}(T_j|T_i,N,\mathbf{x}) > \mathbb{P}(T_j|N, \mathbf{x})$, as the occurrence of $T_i$ already establishes the presence of $B_{1j}$, increasing the likelihood of the occurrence of $T_j$. This complicates the direct use of the PIE, as it involves a combinatorial evaluation of the conditional probabilities in the chain rule expansion of $S_k$ stated in Theorem~\ref{thm:IEP}. To continue our analysis, we choose to disregard this dependency, and thus, we assume that the events $T_i \sim T$ are i.i.d for all $1 \leq i \leq L$, acknowledging that this simplification incurs a loss in the accuracy of our probabilistic blockage modeling. Consequently, the following approximation is made in the PIE expansion:
\begin{align}
\label{eq:s_k_for_2_blockage}
& \mathbb{P}(D_2|N,\mathbf{x})\approx \!\sum_{k=1}^{L}(-1)^{k+1}\binom{L}{k}\mathbb{P}\!\left(T|N,\mathbf{x}\right)^k.
\end{align}
For $N \geq 3$, we evaluate $\mathbb{P}(T|N, \mathbf{x})$ by utilizing Def.~\ref{def:complete_two_blockage}. More specifically, we require the two disc-shaped blockers $B_1$ and $B_2$ to lie in a region such that $I_A \subset ( I_{B_1} \cup I_{B_2})$, $I_A \not\subset I_{B_1}$ and $I_A \not\subset I_{B_2}$ (see Fig.~\ref{fig:vis-scenarios} (d)). In other words,
\begin{align}
\label{eq:u_k_exact_2_blockage}
\mathbb{P}(T|\mathbf{x}) = p_2(\mathbf{x}) \overset{\Delta}{=} \int_{\mathcal{R}_2(\mathbf{x})} \mathbb{P}(\mathbf{y})\mathbb{P}(\mathbf{z})d\mathbf{y}d\mathbf{z},
\end{align}
where $\mathcal{R}_2(\mathbf{x}) = \{(\mathbf{y}, \mathbf{z}) \in Q\times Q\mid (I(\mathbf{x}) \subset I(\mathbf{y}) \cup I(\mathbf{z})) \wedge (I(\mathbf{x}) \not\subset I(\mathbf{y})) \wedge (I(\mathbf{x}) \not\subset I(\mathbf{z}))\}$. We substitute Eq.~\ref{eq:u_k_exact_2_blockage} into Eq.~\ref{eq:s_k_for_2_blockage}, to see that $\mathbb{P}(D_2|N,\mathbf{x}) \approx 1 - (1 - p_2(\mathbf{x}))^{\binom{N-1}{2}}$. The evaluation of $\mathcal{R}_2(\mathbf{x})$ is discussed in Sec.~\ref{sec:exp_res}.

\subsection{$\mathbb{P}(D_1\cap D_2|N,\mathbf{x})$: At Least One Complete 1-Blockage and One Simultaneous 2-Blockage}
In this section, we analyze the probability that the agent $A$ experiences at least one Complete 1-Blockage and at least one Simultaneous 2-Blockage. This term makes sense only when $N \geq 4$, and it is set to zero otherwise. We define $F_i = O_i \cap D_2$ as the event that $A$ experiences a Complete 1-Blockage due to the $i^{\text{th}}$ blocker and concurrently at least one Simultaneous 2-Blockage. Thus, we have
\begin{align}
    \mathbb{P}(D_1\cap D_2|N,\mathbf{x}) = \mathbb{P}(F_1 \cup F_2 \cup \hdots \cup F_{N-1} | N, \mathbf{x}).
\end{align}
From the independence of agent placement, it is easy to see that the events $F_j$ are identically distributed. We demonstrate the factorization of a generic summand $S_k$ in order to apply Theorem~\ref{thm:IEP}. Using Defs.~\ref{def:complete_one_blockage} and~\ref{def:complete_two_blockage}, we see that a blocker involved in a Complete 1-Blockage cannot concurrently participate in a Simultaneous 2-Blockage. This yields
\begin{align}
\label{eq:pop_reduce_s_k}
&\mathbb{P}\!\left(F_{i_1}\cap \hdots \cap F_{i_k}\!\mid\! N, \mathbf{x}\right) = \mathbb{P}\!\left(O_{i_1}\cap \hdots O_{i_k}\cap D_2\!\mid\! N, \mathbf{x}\right)\nonumber \\
&=\mathbb{P}\!\left(O_{i_1}\cap \hdots \cap O_{i_k}\!\mid\! N, \mathbf{x}\right)\mathbb{P}(D_2\mid O_{i_1}\cap \hdots O_{i_k}, N, \mathbf{x}) \nonumber \\
&=\mathbb{P}\!\left(O|N,\mathbf{x}\right)^k\mathbb{P}(D_2|N-k,\mathbf{x})
\end{align}
Using Eq.~\ref{eq:s_k_exact_1_blockage} and~\ref{eq:u_k_exact_2_blockage}, we approximate $\mathbb{P}(D_1\cap D_2|N,\mathbf{x})$ as
\begin{align*}
& \approx\sum_{k=1}^{N-3}(-1)^{k+1} \binom{N-1}{k}p_1(\mathbf{x})^k\left(1 - (1 - p_2(\mathbf{x}))^{\binom{N-k-1}{2}}\right)\nonumber
\end{align*}

\subsection{$\mathbb{P}(\mathcal{V}|N,\mathbf{x})$: Likelihood of Visibility}
We now integrate the various theoretical components derived in the previous subsections to assess the likelihood that a single agent $A$ is visible at location $\mathbf{x}\in Q$ in a crowd of $N$ disc-shaped agents. We substitute the relevant terms on the right-hand side of Eq.~\ref{eq:vis_n_x} and simplify to obtain:
\begin{align}
\label{eq:vis_n_x_evaluated}
&\mathbb{P}(\mathcal{V} | N, \mathbf{x}) \approx (1 - p_1(\mathbf{x}))^{N-1} + (1 - p_2(\mathbf{x}))^{\binom{N-1}{2}} - 1 \nonumber \\
&\!\!+\!\!\sum_{k=1}^{N-3}\!(-1)^{k+1}\!\binom{N-1}{k}p_1(\mathbf{x})^k\!\left(\!1\!-\!(1\!-\!p_2(\mathbf{x}))^{\binom{N-k-1}{2}}\!\right)
\end{align}

\begin{figure*}
    \centering
    \includegraphics[width=\linewidth]{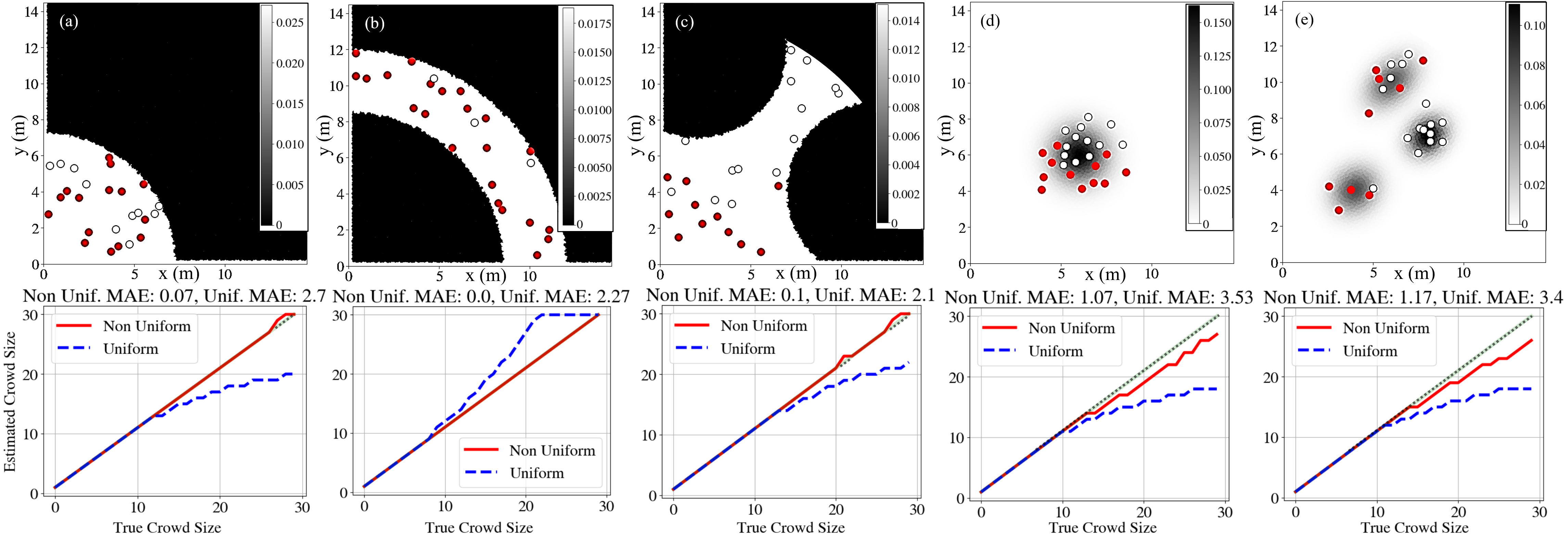}
    \vspace{-18pt}
    \caption{Simulation results for crowd size estimation: (Top) Panels show 5 different sample areas with the non-uniform spatial distribution function, $\mathbb{P}(\mathbf{x})$, displayed in the background. Each panel captures a snapshot of $N$ = 25 agents placed according to $\mathbb{P}(\mathbf{x})$, where filled circles represent visible agents, and empty circles indicate blocked agents for a mmWave TRX located at the origin. (Bottom) Crowd size estimation for different spatial distributions. Lower MAE is observed with our framework (red), particularly for larger crowd sizes, whereas the baseline (blue), which assumes a uniform spatial distribution, leads to less accurate estimation. See color PDF for proper viewing.}
    \vspace{-18pt}
    \label{fig:res_fig}
\end{figure*}

\subsection{Crowd Size Estimation}
We next show how to use the aforementioned models to estimate the crowd size, solely based on the statistics of the observed number of visible agents, denoted by $N_v$ (which is what the mmWave board senses). We define the expected visibility, \textit{i.e.,} the probability that an arbitrary agent is visible in a crowd of size $N$, as $\mathbb{P}(\mathcal{V}|N)$. By averaging our derived $\mathbb{P}(\mathcal{V}|N, \mathbf{x})$ of Eq.~\ref{eq:vis_n_x_evaluated} over the spatial distribution of the crowd, we will have,
\begin{align}
\label{eq:expected_visibility}
    \mathbb{P}(\mathcal{V} | N) = \int_{Q} \mathbb{P}(\mathbf{x})\mathbb{P}(\mathcal{V}|N, \mathbf{x})d\mathbf{x}.
\end{align}

We assume that visibility is independent across the crowd and can thus model the total number of agents visible to the radar as a binomial random variable \textit{i.e.,} $N_v \sim Bin(N, p)$, whose success parameter $p$ is given by $\mathbb{P}(\mathcal{V}\mid N)$. This yields the following analytical expression for the PMF of number of visible agents, given a crowd of size $N$:
\begin{align}
\label{eq:binomial}
    \mathbb{P}_a(N_v|N)=\binom{N}{N_v}(\mathbb{P}(\mathcal{V}|N))^{N_v}(1 - \mathbb{P}(\mathcal{V}|N))^{N-N_v}.
\end{align}
Finally, we can estimate the crowd size by first gathering the statistics of $N_v$ from the mmWave sensing board over a small time interval, and constructing an empirical distribution $\mathbb{P}_e(N_v=n)$. We then estimate $N \in \mathcal{N} = \{1, 2, \hdots, N_{\text{max}}\}$ as the value that minimizes the Kullback-Leibler divergence between $\mathbb{P}_e$ and $\mathbb{P}_a(N_v\mid N)$:
\begin{align}
\label{eq:network_size}
    N^*\!\!=\!\argmin_{N\in \mathcal{N}}D_{\text{KL}}\biggl(\mathbb{P}_e(N_v=n)\ ||\ \mathbb{P}_a(N_v=n\mid N)\biggr).
\end{align}

We next present numerical examples that demonstrate the effectiveness of this crowd size estimation technique.

\section{Numerical Examples and Discussion}\label{sec:exp_res}
In this section, we present simulation-based examples that demonstrate the effectiveness of our proposed crowd size estimation method. Specifically, we provide results for crowds of varying sizes sampled from a non-uniform spatial distribution, $\mathbb{P}(\mathbf{x})$. We set $r_m=14.5$~m, $\rho=0.25$~m, and we realize a crowd of $N$ agents by sampling $N$ times from the points in $Q$ using $\mathbb{P}(\mathbf{x})$. We generate $10,000$ realizations for each $N$ (corresponding to $\sim$~2 minutes of FMCW transmission at a $100$~Hz frame rate) and evaluate the number of visible agents in each realization based on the blockage model described in Sec.~\ref{sec:block_mod}, as illustrated in Fig.~\ref{fig:res_fig} (Top). We then compile the statistics of visible agents and construct the empirical distribution $\mathbb{P}_e(N_v = n)$. Finally, we estimate the optimal $N^*$ using Eq.~\ref{eq:network_size}.

To compute $\mathbb{P}_a(N_v\mid N)$ as derived in Eq.~\ref{eq:binomial}, we need to evaluate the spatial integrals in Eqns.~\ref{eq:s_k_exact_1_blockage}, \ref{eq:u_k_exact_2_blockage}, and \ref{eq:expected_visibility}, which can be computationally demanding. We next outline our efficient numerical integration method. More specifically, these integrals are evaluated using standard Quasi Monte Carlo integration over a 2D Sobol sequence $\mathcal{S} = \{\mathbf{x}_i\in Q\}$~\cite{joe2008constructing}, where all spatial integrals over $Q$ are computed as summations over the points in $\mathcal{S}$. While Theorem~\ref{thm:single_onion} helps define the integration domain for $p_1(\mathbf{x})$ \textit{i.e.,} $\mathcal{R}_1(\mathbf{x})$ (see Eq.~\ref{eq:s_k_exact_1_blockage}), an analytical expression for the domain of $p_2(\mathbf{x})$ \textit{i.e.,} $\mathcal{R}_2(\mathbf{x})$ (see Eq.~\ref{eq:u_k_exact_2_blockage}), presents challenges due to the cardinality of $\mathcal{S}\times \mathcal{S}$. Therefore, we apply Lemma~\ref{lemma:complete_2_blocker} to truncate the search space, focusing only on pairs of blockers with a non-zero overlap in their visibility intervals.

We present sample results for the two categories of spatial inhomogeneity introduced in Fig.~\ref{fig:motivation}. Specifically, Figs.~\ref{fig:res_fig} (a), (b), and (c) depict crowds confined to specific navigable sub-regions of $Q$, resulting in a binary distribution. Notably, Fig.~\ref{fig:res_fig} (b) and (c) have non-convex navigable areas, which are even more challenging to address.  Figs.~\ref{fig:res_fig} (d) and (e), on the other hand, depict crowds aggregating to form hotspots, leading to a more gradual decay in the spatial density.

The results shown in Fig.~\ref{fig:res_fig} (Bottom) demonstrate that our framework can estimate the crowd size well. For comparison, the figures also show the baseline state-of-the-art, which assumes a uniform spatial distribution. As can be seen, our framework consistently performs well across various non-uniform distributions, while assuming a uniform prior results in significant estimation errors, particularly for larger crowds. 
Overall, we achieve an MAE of 0.48 agents across all presented environments, spatial distributions, and crowd sizes of up to (and including) $N_{\text{max}}$ = 30 agents, compared to the benchmark MAE of 2.8 agents when using the uniform prior, \textit{i.e.,} $\mathbb{P}(\mathbf{x}) = \mathds{1}_Q(\mathbf{x})/|Q|$.  

The next steps in this work would be to test the proposed framework in more advanced simulators that account for human body dynamics, temporal crowd evolution, and precise mmWave scattering off of agents, in addition to experimental validations.

\section{Conclusion}
 In this paper, we proposed a novel approach for
crowd size estimation using monostatic mmWave radar. Our framework is
capable of accurately counting large crowds, even when they are not evenly
distributed across space. We first derive location-dependent occlusion probabilities, which we then use to mathematically characterize the probability distributions for the number of agents visible to the radar, as a function of the crowd size. Our numerical simulations demonstrate strong performance, with a mean absolute error of 0.48 agents across various environments and crowd spatial distributions, while handling scenarios with up to and including 30 agents. Overall, this work can enable network resource management as part of ISAC and can further contribute to other applications in the areas of crowd management and urban planning.

\bibliographystyle{IEEEtran}
\bibliography{main}
\end{document}